\documentclass[aps,11pt,twoside, nofootinbib, superscriptaddress]{revtex4}
\usepackage{amsthm}
\usepackage{amsmath,latexsym,amssymb,verbatim,enumerate,graphicx}
\usepackage{subfig}
\usepackage{color}

\newtheorem*{rep@theorem}{\rep@title}
\newcommand{\newreptheorem}[2]{%
\newenvironment{rep#1}[1]{%
 \def\rep@title{#2 \ref{##1} (restatement)}%
 \begin{rep@theorem}}%
{\end{rep@theorem}}}
\makeatother

\newtheorem{thm}{Theorem}
\newtheorem*{thm*}{Theorem}
\newtheorem{prop}[thm]{Proposition}
\newtheorem*{prop*}{Proposition}
\newtheorem{lem}[thm]{Lemma}
\newtheorem*{lem*}{Lemma}

\newtheorem*{fact*}{Fact}
\newtheorem{cor}[thm]{Corollary}
\newtheorem*{cor*}{Corollary}

\newreptheorem{thm}{Theorem}
\newreptheorem{lem}{Lemma}
\newreptheorem{cor}{Corollary}
\newreptheorem{prop}{Proposition}

\def\ba#1\ea{\begin{align}#1\end{align}}
\def\ban#1\ean{\begin{align*}#1\end{align*}}

\newcommand{\ot}{\otimes}
\newcommand{\be}{\begin{equation}}
\newcommand{\ee}{\end{equation}}

\def\benum{\begin{enumerate}}
\def\eenum{\end{enumerate}}

\def\squareforqed{\hbox{\rlap{$\sqcap$}$\sqcup$}}
\def\qed{\ifmmode\squareforqed\else{\unskip\nobreak\hfil
\penalty50\hskip1em\null\nobreak\hfil\squareforqed
\parfillskip=0pt\finalhyphendemerits=0\endgraf}\fi}
\def\endenv{\ifmmode\;\else{\unskip\nobreak\hfil
\penalty50\hskip1em\null\nobreak\hfil\;
\parfillskip=0pt\finalhyphendemerits=0\endgraf}\fi}

\newcommand{\bra}[1]{\langle #1|}
\newcommand{\ket}[1]{|#1\rangle}

\newcommand{\tr}{\text{tr}}

\newcommand{\id}{\mathbb{I}}

\newcommand{\<}{\langle}
\renewcommand{\>}{\rangle}

\def\id{{\operatorname{id}}}

\def\be{\begin{equation}}
\def\ee{\end{equation}}
\def\ben{\begin{eqnarray}}
\def\een{\end{eqnarray}}
\def\ot{\otimes}

\def\bei{\begin{itemize}}
\def\eei{\end{itemize}}

\mathchardef\ordinarycolon\mathcode`\:
\mathcode`\:=\string"8000
\def\vcentcolon{\mathrel{\mathop\ordinarycolon}}
\begingroup \catcode`\:=\active
  \lowercase{\endgroup
  \let :\vcentcolon
  }

\newcommand{\nc}{\newcommand}
 \nc{\proj}[1]{|#1\rangle\!\langle #1 |} 
\nc{\avg}[1]{\langle#1\rangle}

\nc{\conv}{\operatorname{conv}}
\nc{\smfrac}[2]{\mbox{$\frac{#1}{#2}$}} \nc{\Tr}{\operatorname{Tr}}
\nc{\ox}{\otimes} \nc{\dg}{\dagger} \nc{\dn}{\downarrow}
\nc{\lmax}{\lambda_{\text{max}}}
\nc{\lmin}{\lambda_{\text{min}}}

\nc{\csupp}{{\operatorname{csupp}}}
\nc{\qsupp}{{\operatorname{qsupp}}} \nc{\var}{\operatorname{var}}
\nc{\rar}{\rightarrow} \nc{\lrar}{\longrightarrow}
\nc{\poly}{\operatorname{poly}}
\nc{\polylog}{\operatorname{polylog}} \nc{\Lip}{\operatorname{Lip}}
\nc{\Om}{\Omega}
\nc{\wt}[1]{\widetilde{#1}}

\def\>{\rangle}
\def\<{\langle}

\nc{\glneq}{{\raisebox{0.6ex}{$>$}  \hspace*{-1.8ex} \raisebox{-0.6ex}{$<$}}}
\nc{\gleq}{{\raisebox{0.6ex}{$\geq$}\hspace*{-1.8ex} \raisebox{-0.6ex}{$\leq$}}}

\nc{\vholder}[1]{\rule{0pt}{#1}}
\nc{\wh}[1]{\widehat{#1}}
\nc{\h}[1]{\widehat{#1}}

\nc{\ob}[1]{#1}

\def\beq{\begin {equation}}
\def\eeq{\end {equation}}

\def\be{\begin{equation}}
\def\ee{\end{equation}}

\nc{\eq}[1]{(\ref{eq:#1})} 
\nc{\eqs}[2]{\eq{#1} and \eq{#2}}

\nc{\eqn}[1]{Eq.~(\ref{eqn:#1})}
\nc{\eqns}[2]{Eqs.~(\ref{eqn:#1}) and (\ref{eqn:#2})}

\nc{\region}{\cS\cW}

\newenvironment{protocol*}[1]
  {
    \begin{center}
      \hrulefill\\
      \textbf{#1}
  }
  {
    \vspace{-1\baselineskip}
    \hrulefill
    \end{center}
  }

\begin{document}

\title{Generic emergence of classical features in quantum Darwinism}

\author{Fernando G.S.L. Brand\~{a}o}
\affiliation{Quantum Architectures and Computation Group, Microsoft Research, Redmond, WA 98052, USA}
\affiliation{Department of Computer Science, University College London}

\author{Marco Piani}
\affiliation{Department of Physics\&Astronomy and Institute for Quantum Computing \\  University of Waterloo, Waterloo, Ontario, N2L 3G1, Canada}
\affiliation{SUPA and Department of Physics, University of Strathclyde, Glasgow G4 0NG, UK}

\author{Pawe{\l} Horodecki}
\affiliation{National Quantum Information Center of Gda\'{n}sk, 81-824 Sopot, Poland}
\affiliation{Faculty of Applied Physics and Mathematics, Technical University of Gda\'{n}sk, 80-233 Gda\'{n}sk, Poland}

\date{\today}

\begin{abstract}
Quantum Darwinism explains the emergence of classical reality from the underlying quantum reality  by the fact that a quantum system is observed indirectly, by looking at parts of its environment, so that only specific information about the system that is redundantly proliferated to many parts of the environment becomes accessible and objective. However it is not clear under what conditions this mechanism holds true.  Here we rigorously prove that the emergence of classicality is  a general feature of any quantum dynamics:  observers who acquire information about a quantum system indirectly have access at most to classical information about one and the same measurement of the quantum system; moreover, if such information is available to many observers, they necessarily agree. Remarkably, our analysis goes beyond the system-environment categorization. We also provide a full characterization of the so-called quantum discord in terms of local redistribution of correlations.
\end{abstract}

\maketitle

Our best theory of the fundamental laws of physics, quantum mechanics, has counter-intuitive features that are not directly 
observed in our everyday classical reality (e.g., the superposition principle, complementarity, and non-locality). Furthermore, the postulates of quantum mechanics reserve a special treatment to the act of observation, which contrary to its classical counterpart is not a passive act. 
The following fundamental questions then naturally emerge: Through what process does the quantum information contained in a quantum system become classical to an observer? And how come different observers agree on what they see?

The issues of the so-called \emph{quantum-classical boundary} and of the related \emph{measurement problem} dominated large part of the discussions of the early days of quantum mechanics. Indeed, the debate between Bohr and Einstein on the meaning and correctness of quantum mechanics often revolved around the level where quantum effects would disappear---ranging from the microscopic system observed, up to the observer himself. From a practical perspective, our ability to manipulate quantum systems preserving their quantum features has made enormous progresses in recent years---enough to purportedly lead A. Zeilinger to state that ``the border between classical and quantum phenomena is just a question of money"~\cite{AtoZ}.
However, even if we are somewhat pushing the location of the quantum-classical border thanks to our increased experimental ability, a fully satisfactory analysis of the quantum-to-classical transition is still lacking. Such an analysis would both deepen our understanding of the world and conceivably lead to improved technological control over quantum features.

Substantial progress towards the understanding of the disappearance of quantum features was made through the study of \textit{decoherence} \cite{Zur03,Joos03}, where information is lost to an environment.
This typically leads to the selection of persistent \textit{pointer states} \cite{Zur03}, while superpositions of such pointers states are suppressed. Pointer states---and convex combinations thereof---then become natural candidates for classical states. However decoherence by itself does not explain how information about the pointer states reaches the observers, and how such information becomes objective, i.e. agreed upon by several observers. A possible solution to these questions comes from an intriguing idea termed \textit{quantum Darwinism} \cite{Zur09, ZQZ09, BHZ08, BHZ06, BKZ05, OPZ05, OPZ04, RZZ12, RZ11, RZ10, ZQZ10, ZZ13, korbicz1,korbicz2}, which promotes the environment from passive sink of coherence for a quantum system to the active carrier of information about the system (see Figure~\ref{fig:qdarwisnism}). In this view, pointer observables correspond to information about a physical system that the environment---the same environment responsible for decoherence---selects and proliferates, allowing potentially many observers to have access to it. 

\begin{figure}%
\begin{center}
\subfloat[]{\includegraphics[scale=0.35,trim = 0cm -0.2cm 0cm 0cm]{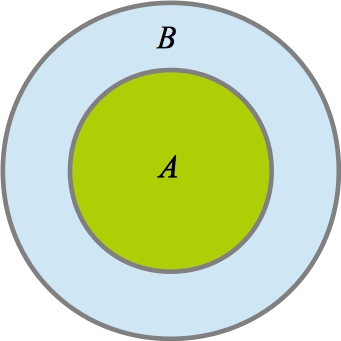}}\hspace{0.08\textwidth}
\subfloat[]{ \includegraphics[scale=0.35, trim = 0cm 0cm 0cm 0cm]{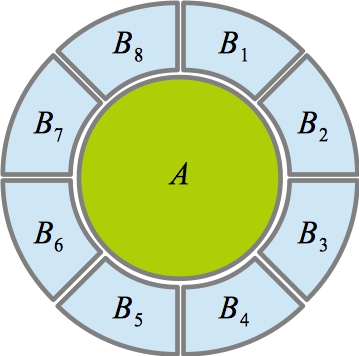}}\\
\subfloat[]{ \includegraphics[scale=0.35]{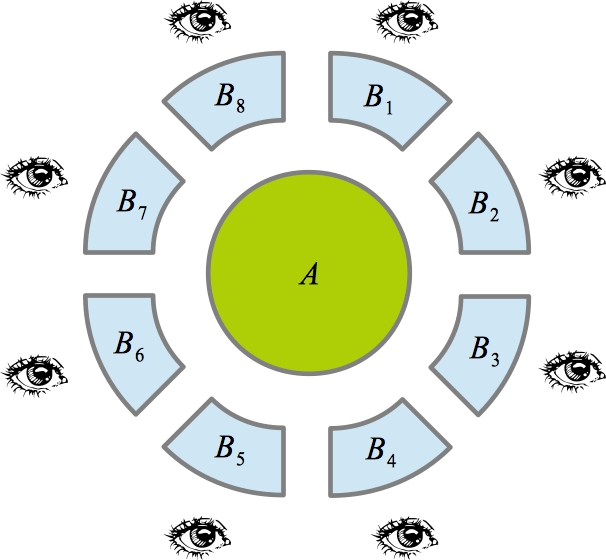}}
\end{center}
\caption{The mechanism for the emergence of objectivity known as quantum Darwinism: the environment as carrier of information. (a) The disappearance of quantum coherence in a system $A$ can be explained in terms of decoherence induced by the interaction with an environment $B$. (b) The environment $B$ responsible for decoherence can be thought as being made up of several parts $B_1$, $B_2$, ..., $B_n$. (c) Observers have indirect access to (the information about) system $A$ through their ability to interact with the environment. Each observer is expected to be able to probe only some part of the environment. Only information about the system that is proliferated in the many parts of the environment is effectively available to the observers, and is necessarily classical and objective.}
\label{fig:qdarwisnism}
\end{figure}

The ideas of quantum Darwinism are beautiful and physically appealing. Significant progress was achieved in a sequence of papers \cite{Zur09, ZQZ09, BHZ08, BHZ06, BKZ05, OPZ05, OPZ04, RZZ12, RZ11, RZ10, ZQZ10, ZZ13, korbicz1, korbicz2}. However we are still far from understanding how generally the ideas of quantum Darwinism apply.
In particular, for example, given any specific interaction Hamiltonian it is not clear whether and to 
what extent classicality sets in. A careful and far-from-trivial analysis must in principle be separately performed for each specific model  (see the papers cited above). So an important question is: Suppose we do not know  {\it anything} about the interaction between the system and its environment; can we still expect some emergence of classicality?
As we shall see below, the answer is positive. Surprisingly, in our analysis what matters 
is just the relation of the Hilbert space dimensionality of an elementary 
subsystem to the number of all the subsystems involved in the interaction, with no dependence on any detail of the dynamics. Thus, one main consequence of our results, which are very general but still derived in full mathematical rigor through information-theoretic techniques, is a deep qualitative change in the study of the emergence of classicality: from proving it in given 
models to showing that it is present in some specific sense (see below) in {\it any} model 
involving sufficiently many subsystems of discrete variables.

We remark that we prove that quantum Darwinism applies 
beyond the system-environment categorization: in a global system composed of many initially uncorrelated subsystems, any subsystem is being objectively measured 
by the other ones (see Fig. \ref{fig:manysystems}).
\begin{figure}
\begin{center}
\subfloat[]{\includegraphics[scale=0.35]{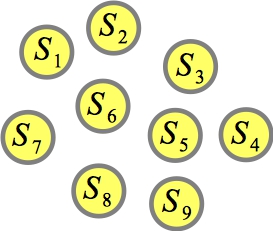}}\hspace{0.1\textwidth}
\subfloat[]{ \includegraphics[scale=0.35]{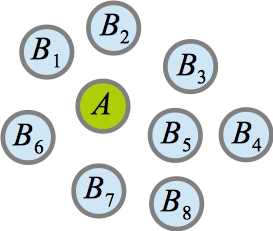}}\hspace{0.1\textwidth}
\subfloat[]{ \includegraphics[scale=0.35]{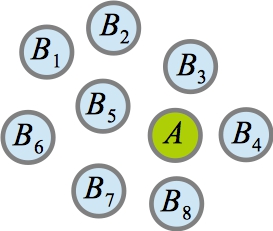}}
\end{center}
\caption{Changing perspective: quantum Darwinism beyond the system-environment categorization. (a) We consider the case where we deal with many systems $S_1, ..., S_m$ that are initially uncorrelated. (b)-(c) The role of the system of interest $A$ can be played indifferently by any subsystem $S_i$, with the remaining subsystems playing the role of the 
(elements of the) environment $B$. As it can be seen in the Results part of the text our results do not depend on any assumed physical symmetry: they are ``symmetric'' themselves and can be applied for any choice of assignment ``system-environment'' and any global interaction. We can conclude that any system is being ``objectively measured'' by the other systems.}
\label{fig:manysystems}
\end{figure}
Most importantly, our approach allows to exactly identify which aspects of emergent objectivity are independent from the specific evolution/interaction, and which do instead depend on the model.  Indeed the present analysis splits the concept of emergent objectivity into two elements:

\begin{itemize}
  \item  (\textit{objectivity of observables}) Observers that access a quantum system by probing part of the environment of the system can only learn about the measurement of a \textit{preferred observable} (usually associated to a measurement on the pointer basis determined by the system-environment interaction \cite{Zur09}). The preferred observable should be independent of which part of the environment is being probed. 
  \item (\textit{objectivity of outcomes}) Different observers that access different parts of the environment have 
(close to) full access to the information about the preferred observable and will agree on the outcome obtained 
(cf. the agreement condition of Ref. \cite{korbicz1}).  
\end{itemize}

The two properties above ensure that the information about the quantum system becomes objective,  
being accessible simultaneously to many observers, and agreed upon. As we report in the Results, 
the first aspect of objectivity---objectivity of observables---is \emph{always} present, i.e., any subsystem is objectively measured by 
the others. On the other hand, the validity of the objectivity of outcomes depends on how much knowledge about the preferred observable is available to the elementary subsystems.

Finally, we make use of our techniques to prove in full generality (i.e., going well beyond the pure-state case treated in~\cite{SZ13}) that when information is distributed to many parties, the minimal average loss in correlations is equal to the quantum discord~\cite{review}, a quantity that has recently attracted much attention but was still missing a full clear-cut operational characterization.

\section{Results}

\subsection{Physical motivation and notation}
\label{sec:motivationnotation}

We want to analyze how the quantum information content of a physical system spreads to (many parts of) its environment. To model this, although our mathematical description in terms of quantum channels (see shortly below) allows for a more general scenario, consider $n+1$ systems $S_1, \ldots, S_{n+1}$ (see Fig.~\ref{fig:manysystems}(a)). These may  constitute a closed system, or be part of a larger system. We focus our attention on one system $S_i$, which we shall call $A$ (see, e.g., Fig.~\ref{fig:manysystems}(b)), and we think of the others systems, now denoted $B_1,\ldots,B_n$, as of fragments of its environment. All our results assume that $A$ is finite-dimensional, with dimension $d_A$, but we do not need such an assumption for the systems $B_1,\ldots,B_n$. Suppose that $A$ is initially decorrelated from $B_1,\dots,B_n$. Independently of any detail of the closed (that is, unitary) or open dynamics of $S_1,\ldots,S_{n+1}$, this condition ensures that the effective transfer of quantum information from $A$ to $B_1, \ldots, B_n$ is represented by a quantum channel  (also called a quantum operation)---a completely positive trace-preserving (cptp) map---$\Lambda : {\cal D}(A) \rightarrow {\cal D}(B_1 \otimes \ldots \otimes B_n)$, with ${\cal D}(X)$ the set of density matrices over the Hilbert space $X$ (see Figure~\ref{fig:dynamics})~\cite{nielsenchuang}. We remark that the role of $A$ can be taken by any $S_i$, as long as it satisfies the condition of being finite-dimensional and initially uncorrelated from the other systems (compare Fig.~\ref{fig:manysystems}(b)) and Fig.~\ref{fig:manysystems}(c)).

\begin{figure}%
\centering
\includegraphics[scale=0.35]{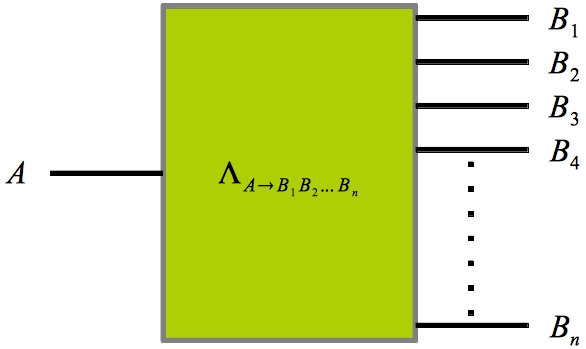}
\caption{The interaction with the environment as quantum channel. The transfer of information from a quantum system $A$ to the many parts $B_1B_2\ldots B_n$ of the environment can be described by a quantum channel, i.e. a completely positive trace-preserving (cptp) map $\Lambda$. Information flows from the left to the right.}
\label{fig:dynamics}
\end{figure}

Given two quantum operations $\Lambda_1$ and $\Lambda_2$, the diamond norm of their difference is defined as $\Vert \Lambda_1 - \Lambda_2 \Vert_{\Diamond}:= \sup_{X} \Vert (\Lambda_1 - \Lambda_2) \otimes \id (X) \Vert_1 / \Vert X \Vert_1$, with the trace norm $\Vert X \Vert_1 := \tr((X^{\cal y}X)^{1/2})$.  The diamond norm $\Vert \Lambda_1 - \Lambda_2 \Vert_{\Diamond}$ gives the optimal bias of distinguishing the two operations $\Lambda_1, \Lambda_2$ by any process allowed by quantum mechanics (i.e. choosing the best possible initial state of the system of interest and of an ancilla system, applying one of the quantum operations to the first system, and performing the best possible measurement to distinguish the two possibilities)~\cite{watrouslectures}. Thus if $\Vert \Lambda_1 - \Lambda_2 \Vert_{\Diamond} \leq \varepsilon$, the two maps represent the same physical dynamics, up to error $\varepsilon$.
Finally, let $\tr_{\backslash X}$ be the partial trace of all subsystems except $X$.

\subsection{Objectivity of Observables}

Our main result is the following (see Figure~\ref{fig:main}):

\begin{figure}%
\centering
\subfloat[]{ \includegraphics[scale=0.35]{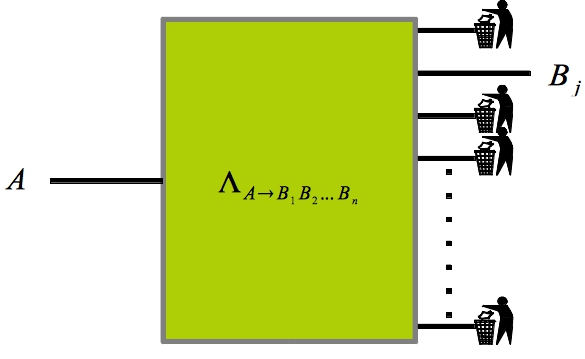}}\quad
\subfloat[]{ \includegraphics[scale=0.35]{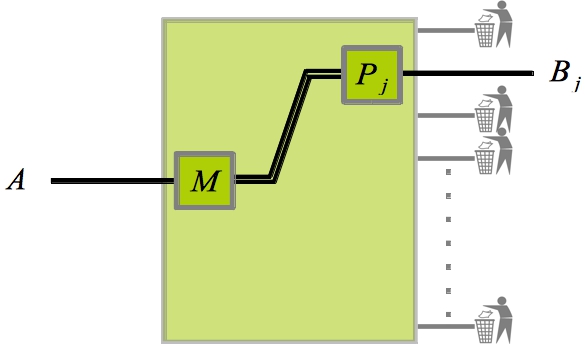}}
\caption{The main result. (a) The mapping from a system $A$ to the many parts $B_1B_2\ldots B_n$ of the environment induces an effective map from $A$ to each part of the environment $B_j$, corresponding to tracing out (i.e., ``throwing away'') the rest of the environment. (b) For most of the effective maps $A\rightarrow B_j$, the dynamics can be well approximated by a measure-and-prepare quantum channel, i.e. by a process where the results of a measurement $M$ on the input are used to decide which output to create at  a later preparation stage $P_j$. A key point  that we prove is that, while the preparation process depends on which part $B_j$ of the environment one considers (symbolized by the ``$j$'' in $P_j$), the measurement stage is independent of it. This implies that only classical information, and only about a specific measurement on $A$, is at best accessible to each observer who can only probe a fragment $B_j$ of the environment. Single lines indicate quantum information (qubits); double lines correspond to classical information (bits). Information flows from left to right.}
\label{fig:main}
\end{figure}

\def\mainthm{

Let $\Lambda : {\cal D}(A) \rightarrow {\cal D}(B_1 \otimes \ldots \otimes B_n)$ be a cptp map. Define $\Lambda_j :=  \tr_{\backslash B_j} \circ \Lambda$ as the effective dynamics from ${\cal D}(A)$ to ${\cal D}(B_j)$ and fix a number $\delta > 0$. Then there exists a measurement, described by a positive-operator-valued measure (POVM) $\{ M_{k} \}_k$  ($M_{k} \geq 0$, $\sum_{k} M_{k} = I$~\cite{nielsenchuang}), and a set $S \subseteq \{1, \ldots, n \}$ with $|S| \geq (1-\delta)n$ such that for all $j \in S$,
\begin{equation}
\label{mainbound}
\left \Vert \Lambda_j - {\cal E}_{j}  \right \Vert_{\Diamond} \leq \left(  \frac{27 \ln(2) (d_{A})^6 \log(d_A)}{n\delta^3} \right)^{1/3}, 
\end{equation}
with
\begin{equation}
\label{measureandprepare}
{\cal E}_{j}(X) := \sum_{k} \tr(M_{k} X) \sigma_{j, k},
\end{equation}
for states $\sigma_{j, k} \in {\cal D}(B_j)$. Here $d_A$ is the dimension of the space $A$.

}
\begin{thm}\label{main}
\mainthm
\end{thm}

As we mentioned before, the diamond-norm distance on the left-hand side of Eq. \eqref{mainbound} represents how different the two physical processes $\Lambda_j$ and ${\cal E}_{j}$ are: the smaller the diamond norm, the more similar the processes, to the extent that they can become indistinguishable. The right-hand side of Eq. \eqref{mainbound} is a bound on such a distinguishability that 
for fixed $\delta$---or even for $\delta$ decreasing with $n$ but not too fast, e.g., for $\delta=n^{-\frac{1-\eta}{3}}$, for any $0<\eta<1$---becomes smaller and smaller as $n$ increases. So, for fixed $d_A$, in the case where we consider an environment with a large number of parts $n$ (e.g., $10^{15}$), for all environment parts but $\delta n$ of them the bound on the right-hand side of Eq. \eqref{mainbound} is very close to zero, i.e. the effective dynamics is ${\cal E}_{j}$ for all practical purposes. 

The operation ${\cal E}_{j}$ in Eq. (2) is termed a \textit{measure-and-prepare} map, since it can be implemented by first measuring the system with the POVM $\{ M_{k} \}_k$ and then preparing a state $\sigma_{j, k}$ depending on the outcome obtained~\cite{entanglementbreaking}. It is clear that an observer that has access to ${\cal E}_{j}(\rho)$ can at most learn about the measurement of the POVM $\{ M_{k} \}_k$ on $\rho$ (but possibly not even that if the states $\{ \sigma_{j, k} \}_k$ are not well distinguishable).

A key aspect of the theorem is that the measurement $\{ M_{k} \} _k$ is \textit{independent} of $j$. In words, the theorem says that the effective dynamics from $A$ to $B_j$, for almost all $j \in \{1, \ldots, n \}$, is close to a measure-and-prepare channel ${\cal E}_{j}$, with the associated measurement $\{ M_{k} \} _k$ the \textit{same} for all such $j$. From the perspective of single observers, the evolution $\Lambda$ is well approximated by a measurement of $A$, followed by the distribution of the classical result, which is finally ``degraded'' by a \emph{local} encoding that, for each $B_j$, produces a quantum state $\sigma_{j,k}$ upon receiving the result $k$. 

Therefore the first feature of quantum Darwinism (objectivity of observables) is completely general!  We can interpret $\{ M_k \}_k$ as the \emph{pointer observable} of the interaction $\Lambda$. Note also that the bound is independent of the dimensions of the $B$ subsystems, being therefore very general. Note, however, the dependence on the dimension $d_A$ of the system $A$. Although the functional form of this dependence might be improved, it is clear that no bound independent of $d_A$ can exist. Indeed, suppose $A = A_1, \ldots, A_n$ and consider the noiseless channel from $A$ to $B_1, \ldots, B_n$. It is clear that a dimension-independent statement of the theorem would fail.

\subsection{Objectivity of Outcomes}

We note that Theorem \ref{main} does not say anything about the second part of quantum Darwinism, namely objectivity of outcomes. It is clear that in full generality this latter feature does not hold true. Indeed, as observed already in Ref. \cite{RZZ12}, if $\Lambda$ is a Haar random isometry from $A$ to $B_1, \ldots, B_n$, then for any $i$ for which $B_i$ has less than half the total size of the environment, the effective dynamics from $A$ to $B_{i}$ will be very close to a completely depolarizing one, mapping any state to the maximally mixed state. Therefore objectivity of outcomes must be a consequence of the special type of interactions we have in nature, instead of a consequence of the basic rules of quantum mechanics (in contrast, Theorem \ref{main} shows that objectivity of observables \textit{is} a consequence only of the structure of quantum mechanics).

Can we understand better the conditions under which objectivity of outcomes holds true? First let us present a strengthening of Theorem \ref{main}, where we consider subsets of the environment parts. Let $[n] := \{1, \ldots, n \}$.

\def\cormany{

Let $\Lambda : {\cal D}(A) \rightarrow {\cal D}(B_1 \otimes \ldots \otimes B_n)$ be a cptp map. For any subset $S_t \subseteq [n]$ of $t$ elements, define $\Lambda_{S_t} :=  \tr_{\backslash \cup_{l \in S_t} B_l} \circ \Lambda$ as the effective channel from ${\cal D}(A)$ to ${\cal D}(\bigotimes_{l \in S_t} B_l)$. Then for every $\delta > 0$ there exists a measurement $\{ M_{k} \}_k$  ($M_{k} \geq 0$, $\sum_{k} M_{k} = I$) such that for more than a $(1-\delta)$ fraction of the subsets $S_t \subseteq [n]$,
\begin{equation}
\left \Vert \Lambda_{S_t} - {\cal E}_{S_t}  \right \Vert_{\Diamond} \leq \left(  \frac{27 \ln(2) (d_{A})^6 \log(d_A)t }{n\delta^3} \right)^{1/3}, 
\end{equation}
with
\begin{equation}
{\cal E}_{S_t}(X) := \sum_{k} \tr(M_{k} X) \sigma_{S_t, k},
\end{equation}
for states $\sigma_{S_t, k} \in {\cal D}(\bigotimes_{l \in S_t} B_l)$.

}

\begin{thm} \label{cormany}
\cormany
\end{thm}

Theorem \ref{cormany} says that the effective dynamics to $B_{j_1}, \ldots B_{j_t}$ is close to a measure-and-prepare channel, for most groups of parts of the environment $(j_1, \ldots, j_t)$.
Let us discuss the relevance of this generalization to the objectivity of outcomes question.

Let $B_{j_1}, \ldots, B_{j_t}$ be a block of sites such that the effective dynamics from $A$ to $B_{j_1}, \ldots, B_{j_t}$ is well approximated by
\begin{equation} \label{channelE}
{\cal E}(X) := \sum_{k} \tr(M_{k} X) \sigma_{B_{j_1}, \ldots, B_{j_t}, k},
\end{equation}
for the pointer POVM $\{ M_{k} \}_k$ and states $\{ \sigma_{B_{j_1}, \ldots, B_{j_t}, k}\}_k$. From Theorem \ref{cormany} we know that this will be the case for most of the choices of $B_{j_1}, \ldots, B_{j_t}$. As we mentioned before, for many $\Lambda$ the information about the pointer observable is hidden from any small part of the environment and thus outcome objectivity fails. Suppose however that the $t$ observers having access to $B_{j_1}, \ldots, B_{j_t}$ \textit{do have} close to full information about the pointer observable.  We now argue that this assumption implies objectivity of outcomes.

 To formalize it we consider the guessing probability of an ensemble $\{ p_i, \rho_i \}$ defined by
\begin{equation}
p_{\text{guess}}(\{ p_i, \rho_i  \}) := \max_{ \{ N_i \} } \sum_i p_i \tr(N_i \rho_i),
\end{equation}
where the maximization is taken over POVMs $\{ N_{i} \}_i$. If the probability of guessing is close to one, then one can with high probability learn the label $i$ by measuring the $\rho_i$'s. We have

\begin{prop}  \label{guessin}
Let ${\cal E}$ be the channel given by Eq. (\ref{channelE}). Suppose that for every $i = \{1, \ldots, t\}$ and $\delta > 0$,
\begin{equation} \label{availabilitymain}
\min_{\rho \in {\cal D}(A)}  p_{\text{guess}}( \{ \tr(M_k \rho), \sigma_{B_{j_i}, k} \}    ) \geq 1 - \delta.
\end{equation}
Then there exists POVMs $\{ N_{B_{j_1}, k} \}_k, \ldots, \{ N_{B_{j_t}, k} \}_k$ such that
\begin{equation}
\min_{\rho}  \sum_{k} \tr(M_k \rho)  \tr \left(  \left( \bigotimes_{i} N_{B_{j_{i}}, k} \right) \sigma_{B_{j_1} \ldots B_{j_t}, k}   \right) \geq 1 - 6 t \delta^{1/4}.
\end{equation}
\end{prop}

Eq.~(7) is equivalent to saying that the information about the pointer-observable $\{ M_k \}_k$ is available to each $B_{j_i}$, $i \in \{1, \ldots, t \}$. Assuming the validity of Eq. (7), the proposition shows that if observers on $B_{j_1}, \ldots, B_{j_t}$ measure independently the POVMs $\{ N_{B_{j_1}, k} \}_k, \ldots, \{ N_{B_{j_t}, k} \}_k$, they will with high probability observe the same outcome. Therefore, while objectivity of outcomes generally fails, we see that whenever the dynamics is such that the information about the pointer observable is available to many observers probing different parts of the environment, then they will agree on the outcomes obtained.

\subsection{Deriving Quantum Discord from Natural Assumptions}

Let us now turn to a different consequence of Theorem \ref{main}. In the attempt to clarify and quantify how quantum correlations differ from correlations in a classical scenario, Ollivier and Zurek~\cite{OZ01} (see also \cite{HV01}) defined the discord of a bipartite quantum state $\rho_{AB}$ as
\begin{equation}
\label{eq:discord}
D(A|B)_{\rho_{AB}}:= I(A:B)_{\rho} - \max_{\Lambda \in \text{QC}} I(A:B)_{\id \otimes \Lambda(\rho)},
\end{equation}
where $I(A:B)_\rho = H(A)_\rho + H(B)_\rho - H(AB)_\rho$ is the mutual information, $H(X)_\rho=H(\rho_X)=-\tr\big(\rho_X\log\rho_X\big)$ is the von Neumann entropy, and the maximum is taken over quantum-classical (QC) channels $\Lambda(X) = \sum_{k} \tr(M_k X) \ket{k}\bra{k}$, with a POVM $\{ M_k \}_k$. Notice that Ollivier and Zurek originally~\cite{OZ01} defined discord in terms of projective measurement rather than general POVMs.

The discord quantifies the correlations---as measured by mutual information---between $A$ and $B$ in $\rho_{AB}$ that are inevitably lost if one of the parties (in the definition above, Bob) tries to encode his share of the correlations in a classical system. Alternatively, quantum discord quantifies the minimum amount of correlations lost under local decoherence, possibily after embedding, and in this sense can be linked to the notion of pointer states~\cite{OZ01}. As such, quantum discord is often seen as the purely quantum part of correlations, with the part of correlations that can be transferred to a classical system---alternatively, surviving decoherence---deemed the classical part~\cite{review,OZ01,HV01,nolocalbroadcasting}. 

Recently there has been a burst of activity in the study of quantum discord (see \cite{review}). Despite the recent efforts, the evidence for a clear-cut role of discord in an operational settings is still limited~\cite{review}. Hence it is important to identify situations where discord emerges naturally as the key relevant property of correlations. Here we identify one such setting in the study of the distribution of quantum information to many parties, intimately related to the no-local-broadcasting theorem~\cite{nolocalbroadcasting,luosun2010}. Indeed a corollary of Theorem \ref{main} is the following (see Figure~\ref{fig:discord}):

\def\cordiscord{

Let $\Lambda : {\cal D}(B) \rightarrow {\cal D}(B_1 \otimes \ldots \otimes B_n)$ be a cptp map. Define $\Lambda_j :=  \tr_{\backslash B_j} \circ \Lambda$ as the effective dynamics from ${\cal D}(B)$ to ${\cal D}(B_j)$. Then for every $\delta > 0$ there exists a set $S \subseteq [n]$ with $|S| \geq (1-\delta)n$ such that for all $j \in S$ and all states $\rho_{AB}$ it holds
\begin{equation}
I(A:B_j)_{\id_{A}\ot\Lambda_{j}(\rho_{AB})} \leq\max_{\Lambda \in \text{QC}} I(A:B)_{\id \otimes \Lambda(\rho_{AB})} + \epsilon \log d_A + 2 h_2 \left(\epsilon\right),
\end{equation}
where $\epsilon =  2 \left(  \frac{ 27 \ln(2) (d_{B})^6 \log(d_B)}{n\delta^3} \right)^{1/3}$, $h_2$ is the binary entropy function, and the maximum on the right-hand side is over quantum-classical channels $\Lambda(X) = \sum_{l} \tr(N_{l}X) \ket{l}\bra{l}$, with $\{N_l\}_l$ a POVM and $\{\ket{l}\}_l$ a set of orthogonal states.

As a consequence, for every $\rho_{AB}$,
\begin{equation}
\label{eq:asymptoticbroadcast}
\lim_{n\rightarrow\infty}\max_{\Lambda_{B\rightarrow B_1B_2\ldots B_n}}\mathop{\mathbb{E}}_j I(A:B_j) =  \max_{\Lambda \in \text{QC}}I(A:B)_{\id \otimes \Lambda(\rho_{AB})},
\end{equation}
with $\mathop{\mathbb{E}}_j X_j = \frac{1}{n}\sum_{i=1}^N X_j$, and the maximum on the left-hand side taken over any quantum operation $\Lambda : {\cal D}(B) \rightarrow {\cal D}(B_1 \otimes \ldots \otimes B_n)$.

}

\begin{cor} \label{discord}
\cordiscord
\end{cor}

Therefore we can see the discord of $\rho_{AB}$ as the asymptotic minimum average loss in correlations when one of the parties (Bob, in this case) locally redistributes his share of correlations:
\begin{equation}
\label{eq:discordaverage}
D(A|B)_{\rho_{AB}} = \lim_{n\rightarrow\infty}\max_{\Lambda_{B\rightarrow B_1B_2\ldots B_n}}\mathop{\mathbb{E}}_j \Big(I(A:B)_{\rho_{AB}} - I(A:B_j)_{_{\id \otimes \Lambda(\rho_{AB})}}\Big).
\end{equation}
\begin{figure}%
\centering
\subfloat[]{\includegraphics[scale=0.35]{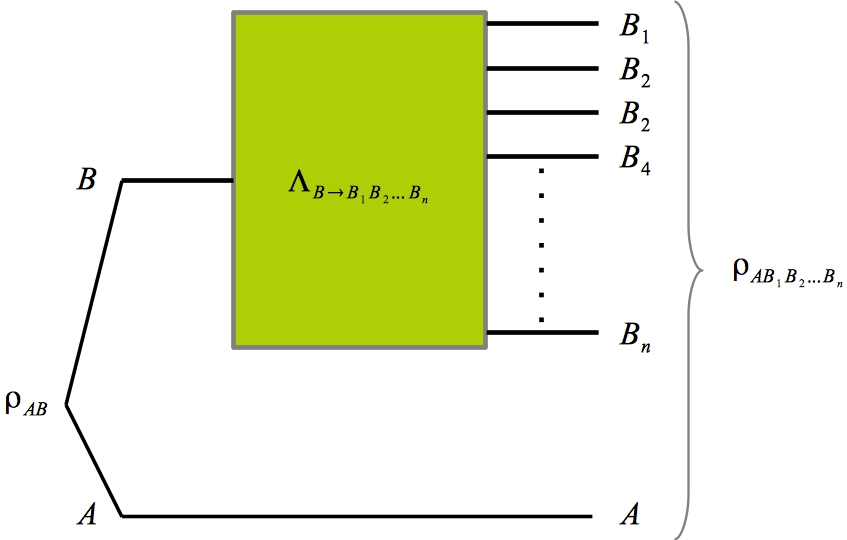}}\\
\subfloat[]{\includegraphics[scale=0.35]{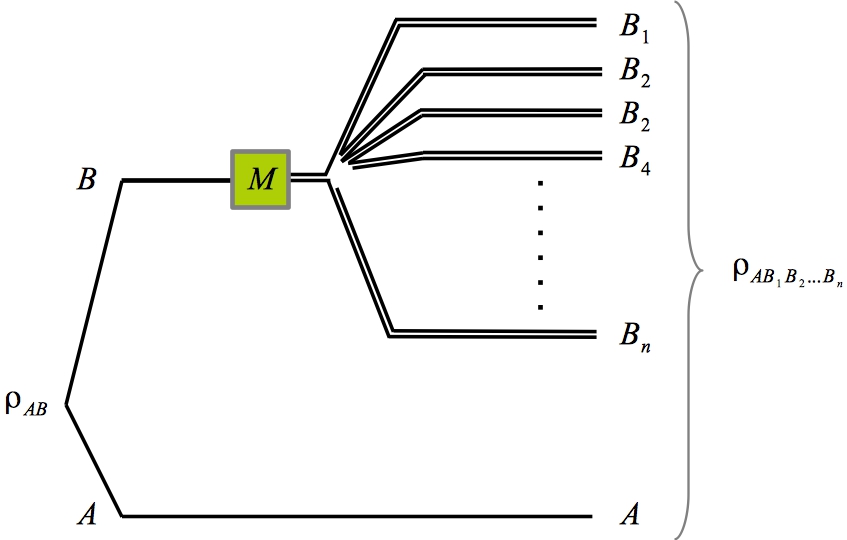}}
\caption{Asymptotic distribution of quantum correlations leads to classicality. (a) The part $B$ of a shared state $\rho_{AB}$ that contains an amount of correlations $I(A:B)$---as measured by mutual information $I$---is split and distributed to many parties $B_1$, $B_2$, ..., $B_n$. We are interested in the largest possible average mutual information $\frac{1}{n}\sum_{j}^n I(A:B_j)$ between $A$ and each $B_j$ after redistribution.  That is, roughly speaking, we want to know what is the best possible way to redistribute correlations so that, on average, each $B_j$ is as correlated with $A$ as possible. In a classical scenario this is trivial: every $B_j$ can be as much correlated with $A$ as the initial $B$, because it always possible to broadcast classical information. (b) We find that, as the number $n$ of recipients $B_j$ grows, the best strategy to redistribute correlations corresponds to reducing it to the classical case in an optimized way. This corresponds to performing the best possible measurement on the original system $B$, followed by the distribution of classical information (the outcome of the measurement) to each $B_j$. Single lines represent quantum information; double lines represent classical information. Information flows from left to right.}
\label{fig:discord}
\end{figure}

Other operational approaches to quantum discord, in particular from a quantum information perspective, have been proposed, but we feel Corollary \ref{discord} stands out in comparison to them. First, Corollary \ref{discord} does not introduce from the start local measurements, which not so surprisingly would lead to the appearance of discord (as per its definition given in Eq.~(\ref{eq:discord})); in contrast, measurements appear as ``effective measurements '' due to the presence of other $B$'s. Second, Corollary \ref{discord} links quantum discord to the the redistribution of quantum systems and quantum correlations in a general and natural way. Notice that this is different from~\cite{cavalcanti11}, where operational interpretations of discord are given that are somewhat more involved, and from~\cite{datta11}, where discord is given an interpretation in quantum communication scenarios that does not really go much beyond its definition. Corollary \ref{discord} also has full validity, applying both to the case where $\rho_{AB}$ is a pure state and when it is mixed. As we will see in Section~\ref{sec:previous}, in particular this removes the limitations of a recent related work by Streltsov and Zurek~\cite{SZ13}.

\subsection{Relation to previous work}
\label{sec:previous}

It is instructive to compare our result to previous work on the subject. In the pioneering works on quantum Darwinism \cite{Zur09, ZQZ09, BHZ08, BHZ06, BKZ05, OPZ05, OPZ04, RZZ12, RZ11, RZ10, ZQZ10, ZZ13}, the focus was on studying specific examples where the emergence of objectivity could be analysed in detail. We regard Theorem \ref{main} as providing a rigorous justification to some of the claims of those works (namely observable objectivity and some aspects of outcome objectivity).

The proliferation of information can intuitively be connected to the idea of cloning of information. The no-cloning theorem \cite{WZ82} is one of the hallmarks of quantum mechanics, stating that only classical information can be perfectly and infinitely cloned. Based on this intuition, in two beautiful papers  first Chiribella and D'Ariano \cite{CD06}  and later Chiribella \cite{ChiriTQC} obtained the closest results to Theorem \ref{main} previously known (building on \cite{BA06, CKMR07}). In those works a variant of Theorem \ref{main} is proven for a dynamics $\Lambda$ in which all the $B$ subsystems are permutation-symmetric, i.e. the information is \textit{symmetrically} distributed in the environment. In particular, bounds similar to Eq. (\ref{mainbound}) were provided, but with the dimension of the $B$ systems in place of the dimension of the $A$ system. Therefore whether the assumption of permutation-symmetry of the $B$ systems (which is hard to justify) was needed, and whether the bound had to depend on the dimensions of the outputs (which limits its applicability), were left as open questions until now.

Corollary \ref{discord} has a similar flavour to a result due to Streltsov and Zurek \cite{SZ13} regarding the role of quantum discord in the redistribution of correlations~\cite{noteoriginal}. However Streltsov and Zurek were only able to  treat  the case where the initial state shared by Alice and Bob is pure. In such a case is was shown that Eq.~\eqref{eq:discordaverage} holds even without the need to consider asymptotics, i.e. without the limit on the right-hand-side of Eq.~(\ref{eq:discordaverage}).

We remark that one can take an alternative approach to the study of the validity of the objectivity conditions of quantum Darwinism, not referring at all to the dynamics---as we instead do in this paper---and rather focusing on the properties of the (final) system-environment state. Such an approach was recently considered in \cite{korbicz1} by asking
what properties the final state of system plus environment should have
to satisfy the conditions of  ``objectivity'' in terms of quantum
measurement theory. It turns out that from a few assumptions, including
Bohr's non-disturbance principle, 
full objectivity requires the so-called {\it broadcast structure}. The latter
has been explicitly shown \cite{korbicz2} to be compatible with what a canonical
physical model involving photon scattering predicts~\cite{RZ11} and with the standard classical
information transmission perspective in terms of accessible information~\cite{ZZ13} (see also Ref. \cite{OHHHH03}
for a general perspective).

\section{Discussion}

The problem of the quantum-to-classical transition---and in particular, the problem of the origin of classical objectivity---is fascinating. The framework of quantum Darwinism appears as an intriguing possible explanation for it. As described in the introduction, quantum Darwinism makes two predictions (which, one could say, constitute its two pillars) on the information about a system that is spread to many observers via the environment that interacts with the system and decoheres it. In this picture, the observers are imagined to acquire information about the system by  each having independent access to some part of the environment.

The first prediction of quantum Darwinism is objectivity of observables, which states that the environment selects the same specific classical information (i.e., information about one specific measurement of the system) to be made potentially available to all the observers. The second prediction is objectivity of outcomes, i.e., the fact that the aforementioned observers will (almost) all have access to the outcome of the observation and agree on it.

The validity and applicability of the quantum Darwinism approach to the problem of the quantum-to-classical transition were so far only partially understood. The fundamental conclusion of quantum Darwinism theory 
\cite{Zur09, ZQZ09, BHZ08, BHZ06, BKZ05, OPZ05, OPZ04, RZZ12, RZ11, RZ10, ZQZ10, ZZ13}
has so far been that the conjunction ``objectivity of measurements \& objectivity of outcomes'' occurs typically in nature because of the specific character of local Hamiltonian interactions. In this work we have rigorously proven that the first pillar of quantum Darwinism---objectivity of observables---is actually completely general, being a consequence of quantum formalism only (in particular of the monogamy of entanglement~\cite{koashiwinter}, but going beyond the latter). That is, objectivity of observables is valid beyond any assumption about the
structure of the interactions. On the other hand, the validity of objectivity of outcomes does seem to depend on the details of the interaction and we are only able to provide partial results about such a feature. Our results seem to indicate that the two pillars of quantum Darwinism are qualitatively different, and suggest that future research should focus on understanding the minimal assumptions needed---within the quantum formalism, which by itself already makes the objectivity of observables a generic feature---to ensure the objectivity of outcomes.

Another striking aspect of the generality of our results is that, as mentioned 
already in the introduction, they actually allow us to go beyond the system-environment categorization.  The key point here is that our analysis does not rely on any 
symmetry assumption about the interaction between the systems $S_1, \ldots, S_{n+1}$ introduced in Section~\ref{sec:motivationnotation}, or about the systems themselves; the conditions of independence and of finite-dimensionality mentioned in Section~\ref{sec:motivationnotation} suffice to ensure that every system is objectively measured by the others. Up to our knowledge this is the first result of this generality. 

A key question is how the present approach can be further generalised 
to an infinite-dimensional system. This will likely require the consideration of bounds of energy and energy fluctuations, leading to the consideration of an effective dimension for physical systems.

Finally, we remark that as a corollary we have also derived a clear-cut operational interpretation to quantum discord, which was originally introduced to capture the quantumness of correlations in information-theoretic terms. We proved that quantum discord corresponds to the asymptotic average loss in mutual information, when one of the parties, e.g. Bob, attempts to distribute his share of the correlations with Alice to many parties. From the perspective of quantum Darwinism, one can interpret this result as the fact that the many observers having each access to only a part of the environment will, on average, only be able to establish at most classical correlations with the system of interest---the system that ``gets measured by the environment''. In this sense, we have fully generalized the results of~\cite{ZZ13} and~\cite{SZ13}, that were limited to pure states.

\section{Methods}

The proofs of Theorems \ref{main} and \ref{cormany}, Proposition \ref{guessin}, and Corollary \ref{discord} are presented in the Supplementary Information. Here we only provide the proof idea of Theorem \ref{main}. It is based on quantum information-theoretic arguments along the lines of recent work by Harrow and one of us \cite{BH12, BH13} for deriving new quantum de Finetti Theorems. We develop the methods of \cite{BH12, BH13} further to show that not only the effective channels ${\cal E}_j$ are close to a measure-and-prepare channel for most $j$, but that the POVM defining the channels is the same for all $j$. This latter feature was not appreciated in \cite{BH12, BH13}, but is fundamental in the context of quantum Darwinism.

The rough idea of the proof is to consider the state obtained by applying the general dynamics on half of a maximally entangled state of the system $A$ and an ancillary system. This gives the state $\rho_{AB_1 \ldots B_n}$ on $AB_1 \ldots B_n$. Then we consider the effect of measuring (in an appropriate basis that must be optimized over and is not given explicitly) a few of the $B_j$ systems of the state $\rho_{AB_1 \ldots B_n}$, for randomly chosen $j'$s. We argue that the statistics of such measurement and the form of the postselected state in system $A$ specifies a POVM $\{ M_k \}_k$ for which Eq. (1) holds true.  This is a consequence of an important property of the quantum mutual information: the chain rule~\cite{nielsenchuang}. Intuitively this process shows that by probing a small part of the environment (with the appropriate measurement) and by considering the effect on the system $A$, the pointer POVM $\{ M_k \}_k$ is fully determined. 

The argument has connections with the phenomenon of entanglement monogamy~\cite{koashiwinter}, which intuitively says that $\rho_{AB_j}$ must be close to a separable state for most $j$. A state $\sigma^{\textrm{sep}}_{AB}$ is separable if it can be written as a convex combination of product states: $\sigma^{\textrm{sep}}_{AB}=\sum_k p_k \sigma_k^A\otimes \sigma_k^B$. Thus, by the Choi-Jamiolkowski isomorphism~\cite{watrouslectures} the associated channel ${\cal E}_j$ must be close to a measure-and-prepare map. But our results go beyond what we simply expect from entanglement monogamy, by showing the existence of the common pointer POVM for most ${\cal E}_j$ (which is equivalent to saying that $\rho_{AB_j}$ is close to $\sum_i p_i \rho_{A, i} \otimes \rho_{B_j, i}$ for an ensemble $\{ p_i, \rho_{A, i} \}$ \textit{independent} of $j$).

\bibliographystyle{apsrev}


\begin{thebibliography}{99}
\bibitem{noteoriginal} The part of our work regarding the interpretation of discord as loss in the local redistribution of correlations originated independently of~\cite{SZ13}. Some ideas and preliminary results were presented already in \cite{Pia12a, Pia12b}
\end{thebibliography}

\section*{Acknowledgements}

FB thanks David Poulin for introducing him to quantum Darwinism and for useful correspondence on the subject. MP thanks Giulio Chiribella, Robert Koening, Jarek Korbicz, Alexander Streltsov, Wojciech Zurek, and Michael Zwolak for useful and stimulating discussions. PH thanks Jarek Korbicz, Ryszard Horodecki and Jess Riedel for discussions on quantum Darwinism. FB was funded by an ESPRC Early Career Fellowship.
MP acknowledges support from NSERC, CIFAR, DARPA, and Ontario Centres of Excellence. PH is supported by the National Science Centre project Maestro DEC-2011/02/A/ST2/00305.

\section*{Author contributions}

All authors contributed extensively to the work presented in this paper.

\section*{Competing financial interests}

The authors declare no competing financial interests.

\section*{SUPPLEMENTARY INFORMATION}


We will make use of the following properties of the mutual information:
\begin{itemize}
\item \emph{Positivity of conditional mutual information}:
\beq
I(A:B|C) := I(A:BC) - I(A:C)\geq 0.
\eeq
This is equivalent to strong subadditivity and to monotonicity of mutual information under local operations~\cite{nielsenchuang}.
\item For a general state $\rho_{AB}$ it holds~\cite{nielsenchuang}
\beq
I(A:B)_{\rho_{AB}}\leq 2\min\{\log d_A, \log d_B\},
\eeq
with the more stringent bound
\beq
I(A:B)_{\sigma^\textrm{sep}_{AB}}\leq  \min\{\log d_A, \log d_B\}
\eeq
for a separable state $\sigma^\textrm{sep}_{AB}$~\cite{kempenielsen}.
\item \emph{Chain rule}~\cite{nielsenchuang}:
\beq
\label{eq:chainrule}
\begin{split}
I(A:B_1B_2\ldots B_n) &=I(A:B_1)+I(A:B_2 | B_1) + I(A:B_3 | B_1B_2)+\ldots\\
&\phantom{=}\ldots+ I(A:B_n|B_1B_2\dots B_{n-1}).
\end{split}
\eeq
\item \emph{Pinsker's inequality} (for mutual information):
\beq
\label{eq:pinsker}
\frac{1}{2 \ln 2}\|\rho_{AB} - \rho_A\otimes \rho_B\|_1^2 \leq I(A:B)_{\rho_{AB}}.
\eeq

\item \emph{Conditioning on classical information}
\begin{equation} \label{conditioning}
I(A : B | Z)_{\rho} = \sum_z p(z) I(A:B)_{\rho_z}
\end{equation}
for a state $\rho_{ABZ} = \sum_z p(z) \rho_{z, AB} \otimes \ket{z}\bra{z}_{Z}$, with $\{\ket{z}\}$ an orthonormal set.
\end{itemize}

\subsection{Proof of Theorem \ref{main}} \label{proofmain}


The first lemma we will use is a variant of Lemma 20 of \cite{BH12b}. 

\begin{lem}
\label{lem:localtracenorm}
Consider a Hermitian matrix $L_{AB} \in \mathbb{B}(\mathbb{C}^{d_A} \otimes \mathbb{C}^{d_B})$, with $d_A \leq d_B$. Then
\[
\| L_{AB} \|_1 \leq d_A^2 \max_{M_B}\|\id_A\otimes M_B\left( L_{AB}  \right)\|_1,
\]
where the maximum is taken over local measurement maps $M_B(Y) = \sum_{l} \tr(N_{l}Y) \ket{l}\bra{l}$ with a POVM $\{ N_l \}$.
\end{lem}
\begin{proof}

Write $L_{AB} =  \sum_{i, j=1}^{d_A} \ket{i}\bra{j} \otimes L_{ij}$ with $\{ \ket{i} \}$ an orthornomal basis for $\mathbb{C}^{d_A}$. On the one hand, thanks to the triangle inequality, we have
\begin{equation}
\Vert L_{AB} \Vert_{1} =  \left \Vert   \sum_{i, j=1}^{d_A} \ket{i}\bra{j} \otimes L_{ij} \right \Vert_1 \leq d_A^2 \max_{i, j} \Vert L_{ij} \Vert_1  
\end{equation}

On the other hand,
\begin{equation}
\begin{split}
&\quad\,\max_{M_B}\|\id_A\otimes M_B\left( L_{AB}  \right)\|_1\\
&= \max_{M_B}\left\|\sum_{i, j=1}^{d_A} \ket{i}\bra{j} \otimes M_B(L_{ij})  \right\|_1\\
&=\max_{M_B}\max_{\|K_{AB}\|\leq1}\left|\tr\left(K_{AB}\left(\sum_{i, j=1}^{d_A} \ket{i}\bra{j} \otimes M_B(L_{ij})  \right)\right)\right|\\
&\geq \max_{M_B}\max_{\substack{K_A=  K_A^\dagger, \|K_A\|\leq1\\
					\|K_{B}\|\leq1}}
					\left|\tr\left(K_A\otimes K_B\left(\sum_{i, j=1}^{d_A} \ket{i}\bra{j} \otimes M_B(L_{ij})  \right)\right)\right| \\
&\geq \max \left\{ \max_{i} \max_{M_B} \Vert M_B(L_{ii}) \Vert_1,  \max_{i \neq j} \max_{M_B} \Vert M_B(L_{ij} + L_{ji}) \Vert_1,   \max_{i \neq j} \max_{M_B} \Vert M_B(i(L_{ij} - L_{ji})) \Vert_1\right\},
\end{split}
\end{equation}
where we have repeatedly used the expression of the trace norm $\|X\|_1= \max_{\|K\|\leq1} |\tr(KX)|$, and the alternative choices $K_A=\ket{i}\bra{i}$ , $K_A=\ket{i}\bra{j}+\ket{j}\bra{i}$, or $K_A=i(\ket{i}\bra{j}-\ket{j}\bra{i})$ to arrive to the last inequality.

It's clear that 
\begin{equation}
\max_{M_B} \Vert M_B(L_{ii}) \Vert_1 = \Vert L_{ii} \Vert_1 
\end{equation}
and similarly
\begin{equation}
\max_{M_B} \Vert M_B(L_{ij} + L_{ji}) \Vert_1 = \Vert L_{ij} + L_{ji} \Vert_1,\qquad \max_{M_B} \Vert M_B(i(L_{ij} - L_{ji})) \Vert_1 = \Vert L_{ij} - L_{ji} \Vert_1 . 
\end{equation}

To complete the proof it is enough to observe
\begin{equation}
\|L_{ij}\|_1\leq \frac{1}{2} (\|L_{ij}+L_{ji}\|_1 +\|L_{ij}-L_{ji}\|_1)\leq  \max\{\Vert L_{ij} + L_{ji} \Vert_1,\Vert L_{ij} - L_{ji} \Vert_1\}.
\end{equation}
\end{proof}



The second lemma bounds the optimal distinguishability of two quantum channels (i.e. their diamond-norm distance) in terms of the distinguishability of their corresponding Choi-Jamio{\l}kowski states.

\begin{lem} \label{lem:choinotbad}
Let $\Phi_{AA'} = d_A^{-1} \sum_{k, k'} \ket{k, k} \bra{k', k'}$ be a $d_A$-dimensional maximally entangled state. For any cptp map $\Lambda : {\cal D}(A) \rightarrow {\cal D}(B)$ we define the Choi-Jamio{\l}kowski state of $\Lambda$ as $J(\Lambda):=\id_A\otimes \Lambda_A(\Phi_{AA'})$. For two cptp maps $\Lambda_0$ and $\Lambda_1$ it then holds
\beq
\label{eq:choinotbad}
\frac{1}{d_A}\|\Lambda_0 - \Lambda_1\|_\Diamond \leq \|J(\Lambda_0) - J(\Lambda_1)\|_1 \leq \|\Lambda_0 - \Lambda_1\|_\Diamond.
\eeq
\end{lem}
\begin{proof}
The second inequality in \eqref{eq:choinotbad} is trivial, as the diamond norm between two cptp maps is defined through a maximization over input states, while $\|J(\Lambda_0) - J(\Lambda_1)\|_1$ corresponds to the bias in distinguishing the two operations $\Lambda_0$ and $\Lambda_1$ by using the maximally entangled state $\Phi_{AA'}$ as input. The first inequality can be derived as follows.

Any pure state $\ket{\psi}_{AA'}$ can be obtained by means of a local filtering of the maximally entangled state, i.e.,
\[
\ket{\psi}_{AA'} = (\sqrt{d_A}C\otimes \openone) \ket{\Phi}_{AA'}
\]
for a suitable $C\in \mathbb{B}(\mathbb{C}^{d_A})$, which, for a normalized $\ket{\psi}_{AA'}$ satisfies $\tr(C^\dagger C) =1$. From the latter condition, we have that $\|C\|_\infty \leq 1$. Let $\ket{\psi}_{AA'}$ be a normalized pure state optimal for the sake of the diamond norm between $\Lambda_0$ and $\Lambda_1$. We find
\[
\begin{split}
\|\Lambda_0 - \Lambda_1\|_\Diamond &= \|\id_A\otimes (\Lambda_0 - \Lambda_1)[\proj{\psi}]\|_1\\
&=\left\|\id_A\otimes (\Lambda_0 - \Lambda_1)\left((\sqrt{d_A}C\otimes \openone) {\Phi}_{AA'}(\sqrt{d_A}C\otimes \openone)^\dagger\right)\right\|_1\\
&=\left\|(\sqrt{d_A}C\otimes \openone) \Big(\id_A\otimes (\Lambda_0 - \Lambda_1)[ {\Phi}_{AA'} ]\Big) (\sqrt{d_A}C\otimes \openone)^\dagger \right\|_1\\
&\leq d_A \|C\|_\infty^2 \| \id_A\otimes (\Lambda_0 - \Lambda_1)[ {\Phi}_{AA'} ]\|_1\\
&\leq d_A \|J(\Lambda_0) - J(\Lambda_1)\|_1,
\end{split}
\]
where we used (twice) H{\"o}lder's inequality $\|MN\|_1\leq \min\{\|M\|_\infty \|N\|_1,\|M\|_1 \|N\|_\infty\}$  in the first inequality, and  $\|C\|_\infty\leq 1$ in the second inequality.
\end{proof}

We are in position to prove the main theorem, which we restate for the convenience of the reader. 

\begin{repthm}{main}
Let $\Lambda : {\cal D}(A) \rightarrow {\cal D}(B_1 \otimes \ldots \otimes B_n)$ be a cptp map. Define $\Lambda_j :=  \tr_{\backslash B_j} \circ \Lambda$ as the effective dynamics from ${\cal D}(A)$ to ${\cal D}(B_j)$ and fix a number $\delta > 0$. Then there exists a measurement $\{ M_{k} \}_k$  ($M_{k} \geq 0$, $\sum_{k} M_{k} = I$) and a set $S \subseteq \{1, \ldots, n \}$ with $|S| \geq (1-\delta)n$ such that for all $j \in S$,
\begin{equation}
\label{mainboundrep}
\left \Vert \Lambda_j - {\cal E}_{j}  \right \Vert_{\Diamond} \leq \left(  \frac{27 \ln(2) (d_{A})^6 \log(d_A)}{n\delta^3} \right)^{1/3}, 
\end{equation}
with
\begin{equation}
\label{measureandpreparerep}
{\cal E}_{j}(X) := \sum_{k} \tr(M_{k} X) \sigma_{j, k},
\end{equation}
for states $\sigma_{j, k} \in {\cal D}(B_j)$. Here $d_A$ is the dimension of the space $A$.

\end{repthm}

\begin{proof}

Let $\Phi_{AA'} = d_A^{-1} \sum_{k, k'} \ket{k, k} \bra{k', k'}$ be a $d_A$-dimensional maximally entangled state and $\rho_{AB_1,\ldots, B_n} := \id_{A} \otimes \Lambda(\Phi_{AA'})$ be the Choi-Jamiolkowski state of $\Lambda$. Define 
$\pi := \id_A \otimes M_{1} \otimes \ldots \otimes M_{n}(\rho)$, for quantum-classical channels $M_1, \ldots, M_n$ defined as $M_i(X) := \sum_{l} \tr(N_{i, l}X) \ket{l}\bra{l}$, for a POVM $\{ N_{i, l} \}_l$.

We will proceed in two steps. In the first we show that conditioned on measuring a few of the $B's$ of $\rho_{AB_1, \ldots, B_n}$, the conditional mutual information of $A$ and $B_i$ (on average over $i$) is small. In the second we show that this implies that the reduced state $\rho_{AB_i}$ is close to a separable state $\sum_z p(z) \rho_{z, A} \otimes \rho_{B_i, z}$, with the ensemble $\{ p(z), \rho_{z, A} \}$ independent of $i$. We will conclude showing that by the properties of the Choi-Jamiolkowski isomorphism, this implies that the effective channel from $A$ to $B_i$ is close to a measure-and-prepare channel with a POVM independent of $i$.

Let $\mu$ be the uniform distribution over $[n]$ and define $\mu^{\wedge k}$ as the distribution on $[n]^k$ obtained by sampling $m$ times without replacement according to $\mu$; i.e.
\be \mu^{\wedge k}(i_1,\ldots,i_k) = 
\begin{cases}
0 & \text{if $i_1,\ldots,i_k$ are not all distinct} \\
\frac{\mu(i_1)\cdots \mu(i_k)}{\sum_{j_1,\ldots,j_k\text{ distinct}} \mu(j_1)\cdots \mu(j_k)}
& \text{otherwise}
\end{cases}
\ee

Then
\begin{eqnarray}  \label{f1}
\log d_{A} &\geq& \mathop{\mathbb{E}}_{ (j_1, \ldots, j_k) \sim \mu^{\wedge k} } \max_{M_{j_1}, \ldots, M_{j_k}} I(A : B_{j_1}, \ldots, B_{j_k})_{\pi} \\
&=&    \mathop{\mathbb{E}}_{ (j_1, \ldots, j_k) \sim \mu^{\wedge k} }   \max_{M_{j_1}, \ldots, M_{j_k}} \left( I(A : B_{j_1})_{\pi} + \ldots + I(A : B_{j_k} | B_{j_1}, \ldots, B_{j_{k-1}})_{\pi} \right) \nonumber \\  &=:& f(k), \nonumber
\end{eqnarray}
The inequality comes from the fact that $\pi$ is separable between $A$ and $B_1B_2\ldots B_n$ because of the action of the quantum-classical channels $M_1, \ldots, M_n$. The second line follows from the chain rule of mutual information given by Eq. (\ref{eq:chainrule}).

Define $J_k := \{ j_1, \ldots, j_{k-1}  \}$. We have
\begin{eqnarray} \label{f2}
f(k) & \stackrel{(i)}{=} &     \mathop{\mathbb{E}}_{ (j_1, \ldots, j_k) \sim \mu^{\wedge k} }  \max_{M_{j_1}, \ldots, M_{j_{k-1}}} \left( I(A : B_{j_1})_{\pi} + \ldots + \max_{M_{j_k}} I(A : B_{j_k} | B_{j_1}, \ldots, B_{j_{k-1}})_{\pi} \right)  \\
&\stackrel{(ii)}{\geq} &    \mathop{\mathbb{E}}_{ (j_1, \ldots, j_{k-1}) \sim \mu^{\wedge k-1} }    \max_{M_{j_1}, \ldots, M_{j_{k-1}}} \mathbb{E}_{j_k \notin J_k} \left( I(A : B_{j_1})_{\pi} + \ldots + \max_{M_{j_k}} I(A : B_{j_k} | B_{j_1}, \ldots, B_{j_{k-1}})_{\pi} \right) \nonumber \\
& \stackrel{(iii)}{=} &      \mathop{\mathbb{E}}_{ (j_1, \ldots, j_{k-1}) \sim \mu^{\wedge k - 1} }    \max_{M_{j_1}, \ldots, M_{j_{k-1}}}  \left( I(A : B_{j_1})_{\pi} + \ldots + \mathbb{E}_{j_k \notin J_k} \max_{M_{j_k}} I(A : B_{j_k} | B_{j_1}, \ldots, B_{j_{k-1}})_{\pi} \right) \nonumber \\
& \stackrel{(iv)}{\geq} &    \mathop{\mathbb{E}}_{ (j_1, \ldots, j_{k-1}) \sim \mu^{\wedge k - 1} }   \max_{M_{j_1}, \ldots, M_{j_{k-1}}}  \left( I(A : B_{j_1})_{\pi} + \ldots +  I(A : B_{j_{k-1}} | B_{j_1}, \ldots, B_{j_{k-2}})_{\pi} \right) \nonumber \\
&+ &    \mathop{\mathbb{E}}_{ (j_1, \ldots, j_{k-1}) \sim \mu^{\wedge k-1} }    \min_{M_{j_1}, \ldots, M_{j_{k-1}}} \mathbb{E}_{j_k \notin J_k} \max_{\Lambda_{j_k}} I(A : B_{j_k} | B_{j_1}, \ldots, B_{j_{k-1}})_{\pi},  \nonumber \\
&\stackrel{(v)}{=}& f(k-1) + \mathop{\mathbb{E}}_{j_1, \ldots, j_{k-1}} \min_{M_{j_1}, \ldots, M_{j_{k-1}}} \mathbb{E}_{j_k} \max_{M_{j_k}} I(A : B_{j_k} | B_{j_1}, \ldots, B_{j_{k-1}})_{\pi},  \nonumber
\end{eqnarray}
where (i) follows since only $I(A : B_{j_k} | B_{j_1}, \ldots, B_{j_{k-1}})_{\pi}$ depends on $M_{j_k}$; (ii)  by convexity of the maximum function; (iii) again because all the other terms in the sum are independent of $j_k$; (iv) directly by inspection and linearity of expectation; and (v) by the definition of $f(k)$ in Eq. (\ref{f1}).



From Eqs. (\ref{f1}) and (\ref{f2}), we obtain
\begin{equation}
\log d_{A} \geq \sum_{q=1}^k   \mathop{\mathbb{E}}_{ (j_1, \ldots, j_{q-1}) \sim \mu^{\wedge q-1} }  \min_{M_{j_1}, \ldots, M_{j_{q-1}}} \mathop{\mathbb{E}}_{j_q \notin J_q} \max_{M_{j_q}} I(A : B_{j_q} | B_{j_1}, \ldots, B_{j_{q-1}})_{\pi},
\end{equation}
and so there exists a $q \leq k$ such that 
\begin{equation}
\mathop{\mathbb{E}}_{ (j_1, \ldots, j_{q-1}) \sim \mu^{\wedge q-1} }  \min_{M_{j_1}, \ldots, M_{j_{q-1}}} \mathop{\mathbb{E}}_{j \notin J_q} \max_{M_{j}} I(A : B_{j} | B_{j_1}, \ldots, B_{j_{q-1}})_{\pi} \leq \frac{\log d_{A}}{k},
\end{equation}
where we relabelled $j_q \rightarrow j$. Thus there exists a $(q-1)$-tuple $J := (j_1, \ldots, j_{q-1})$ and measurements $M_{j_1}, \ldots, M_{j_{q-1}}$ such that
\begin{equation} \label{smallMI}
\mathop{\mathbb{E}}_{j \notin J} \max_{M_{j}} I(A : B_{j} | B_{j_1}, \ldots, B_{j_{q-1}})_{\pi} \leq \frac{\log d_{A}}{k}.
\end{equation}

Let $\rho^z_{AB_{j}}$ be the post-measurement state on $AB_{j}$ conditioned on obtaining $z$ -- a short-hand notation for the ordered collection of the local results -- when measuring $M_{j_1}, \ldots, M_{j_{q-1}}$ in the subsystems $B_{j_1}, \ldots, B_{j_{q-1}}$ of $\rho$. Note that $\rho^z_{A}$ is independent of $B_{j}$ (for $j \notin J$). By Pinsker's inequality \eqref{eq:pinsker}, convexity of $x \mapsto x^2$, and Eq. (\ref{conditioning}),
\begin{eqnarray}
\left \Vert \id_{A} \otimes M_{j} \left( \rho_{AB_{j}} - \mathbb{E}_z \rho^z_{A} \otimes \rho^z_{B_{j}} \right) \right\Vert_1^2 &=& \left \Vert \id_{A} \otimes M_{j} \left( \mathbb{E}_z\rho^z_{AB_{j}} - \mathbb{E}_z \rho^z_{A} \otimes \rho^z_{B_{j}} \right) \right\Vert_1^2
    \nonumber \\
 &\leq& \mathbb{E}_z  \left \Vert  \id_{A} \otimes M_{j} \left( \rho^z_{AB_{j}} - \rho^z_{A} \otimes \rho^z_{B_{j}} \right)  \right\Vert_1^2 \nonumber \\
&\leq&  2\ln(2)  I(A : B_{j} | B_{j_1}, \ldots, B_{j_{q-1}})_{\pi}.
\end{eqnarray}
By Eq. (\ref{smallMI}) and convexity of $x \mapsto x^2$,
\begin{equation} 
\mathop{\mathbb{E}}_{j \notin J} \max_{M_{j}} \left \Vert \id_{A} \otimes M_{j} \left( \rho_{AB_{j}} - \mathbb{E}_z \rho^z_{A} \otimes \rho^z_{B_{j}} \right) \right\Vert_1 \leq \sqrt{2\ln(2)\frac{\log d_{A}}{k}}.
\end{equation}

Now, by Lemma~\ref{lem:localtracenorm}, we have.

\begin{equation}
 \left \Vert  \rho_{AB_{j}} - \mathbb{E}_z \rho^z_{A} \otimes \rho^z_{B_{j}} \right\Vert_1 \leq (d_A)^2 \max_{M_{j}} \left \Vert \id_{A} \otimes M_{j} \left( \rho_{AB_{j}} - \mathbb{E}_z \rho^z_{A} \otimes \rho^z_{B_{j}} \right) \right\Vert_1,
\end{equation}
and so
\begin{equation} 
\mathop{\mathbb{E}}_{j \notin J}  \left \Vert \rho_{AB_{j}} - \mathbb{E}_z \rho^z_{A} \otimes \rho^z_{B_{j}} \right\Vert_1 \leq \sqrt{2\ln(2)\frac{(d_A)^4\log d_{A}}{k}}.
\end{equation}

Note that $\mathbb{E}_z \rho^z_{A} \otimes \rho^z_{B_{j}} = \sum_z p(z) \rho^z_{A} \otimes \rho^z_{B_{j}}$ is the Choi-Jamiolkowski state of a measure-and-prepare channel ${\cal E}_j$~\cite{entanglementbreaking}, since $\mathbb{E}_z \rho^z_{A} = \rho_A = \openone / d_A$. It is explicitly given by
\begin{equation}
{\cal E}_j(X) := d_A \mathbb{E}_{z} \tr((\rho_A^z)^T X) \rho^z_{B_{j}}.
\end{equation} 
Note that the POVM $\{ d_A p(z) \rho_A^z \}$ is independent of $j$. 

Thanks to Lemma~\ref{lem:choinotbad}, we can now bound the distance of two maps by the distance of their Choi-Jamiolkowski states
\begin{equation}
\Vert \tr_{\backslash B_j} \circ \Lambda -  {\cal E}_j \Vert_{\Diamond} \leq d_A \Vert \rho_{AB_j} -  \mathbb{E}_z \rho^z_{A} \otimes \rho^z_{B_{j}}\Vert_{1},
\end{equation}
to find
\begin{equation} 
\mathop{\mathbb{E}}_{j \notin J}  \left \Vert \tr_{\backslash B_j} \circ \Lambda -  {\cal E}_j  \right\Vert_{\Diamond} \leq \sqrt{2\ln(2)\frac{(d_A)^6\log d_{A}}{k}}.
\end{equation}
Then
\begin{eqnarray} \label{expectation}
\mathop{\mathbb{E}}_{j}  \left \Vert \tr_{\backslash B_j} \circ \Lambda -  {\cal E}_j  \right\Vert_{\Diamond} 
&=&  \mathop{\mathbb{E}}_{j \notin J}  \left \Vert \tr_{\backslash B_j} \circ \Lambda -  {\cal E}_j  \right\Vert_{\Diamond}  + \frac{k}{n}\mathop{\mathbb{E}}_{j \in J}  \left \Vert \tr_{\backslash B_j} \circ \Lambda -  {\cal E}_j  \right\Vert_{\Diamond} \nonumber \\ 
&\leq& \sqrt{2\ln(2)\frac{(d_A)^6\log d_{A}}{k}} + \frac{2k}{n},
\end{eqnarray}
where we used that the diamond norm between two cptp maps is upper-bounded by 2. 

Choosing $k$ to minimize the latter bound we obtain~\footnote{The expression $a/\sqrt{k}+b k$ is minimal for $k=(\frac{a}{2b})^{2/3}$. We further use that for $b=2/n<1$ it holds $b^{1/3}\geq b^{5/6}$.}
\begin{equation} \label{expectation2}
\mathop{\mathbb{E}}_{j}  \left \Vert \tr_{\backslash B_j} \circ \Lambda -  {\cal E}_j  \right\Vert_{\Diamond} \leq \left( \frac{27 \ln(2) (d_A)^6 \log(d_A)}{n} \right)^{1/3}.
\end{equation}

Finally applying Markov's inequality,
\begin{equation}
\Pr_{i} \left(  \left \Vert \tr_{\backslash B_j} \circ \Lambda -  {\cal E}_j  \right\Vert_{\Diamond} \geq \frac{1}{\delta} \left( \frac{27 \ln(2) (d_A)^6 \log(d_A)}{n} \right)^{1/3}     \right)  \leq \delta.
\end{equation}
\end{proof}

\subsection{Proof of Theorem \ref{cormany}}

The proof of Theorem \ref{cormany} follows along the same lines as Theorem \ref{main}:

\begin{repthm}{cormany}
\cormany
\end{repthm}

\begin{proof}

Since the proof is very similar to the proof of Theorem \ref{main}, we will only point out the differences. 

Let $\rho_{AB_1,\ldots, B_n} := \id_{A} \otimes \Lambda(\Phi)$ be the Choi-Jamiolkowski state of $\Lambda$ and $C = \{ C_1, \ldots, C_{n/t} \}$ be a partition of $[n]$ into $n/t$ sets of $t$ elements each. Define $\pi_{C} := \id_A \otimes M_{1} \otimes \ldots \otimes M_{n/t}(\rho)$, for quantum-classical channels $M_1, \ldots, M_{n/t}$ defined as $M_i(X) := \sum_{l} \tr(N_{i, l}X) \ket{l}\bra{l}$, for a POVM $\{ N_{i, l} \}_l$, with $M_i$ acting on $\cup_{j \in C_i} B_j$.

As in the proof of Theorem \ref{main}, by the chain rule,
\begin{eqnarray}  \label{f11}
\log d_{A} &\geq&   \mathop{\mathbb{E}}_{C_{j_1}, \ldots, C_{j_k}} \max_{M_{j_1}, \ldots, M_{j_k}} I(A : B_{C_{j_1}}, \ldots, B_{C_{j_k}})_{\pi_{C}} \\
&=&   \mathop{\mathbb{E}}_{C_{j_1}, \ldots, C_{j_k}}   \max_{M_{j_1}, \ldots, M_{j_k}} \left( I(A : B_{C_{j_1}})_{\pi_C} + \ldots + I(A : B_{C_{j_k}} | B_{C_{j_1}}, \ldots, B_{C_{j_{k-1}}})_{\pi_C} \right) =: f(t), \nonumber
\end{eqnarray}
where the expectation is taken uniformly over the choice of non-overlapping sets  $C_{j_1}, \ldots, C_{j_k} \in [n]^t$. 

We have
\begin{eqnarray} \label{f22}
f(t) &=&     \mathop{\mathbb{E}}_{C_{j_1}, \ldots, C_{j_k}}    \max_{M_{j_1}, \ldots, M_{j_{k-1}}}     \left( I(A : B_{C_{j_1}})_{\pi_C} + \ldots + \max_{M_{j_k}} I(A : B_{C_{j_k}} | B_{C_{j_1}}, \ldots, B_{C_{j_{k-1}}})_{\pi_C} \right)  \\
&\geq &        \mathop{\mathbb{E}}_{C_{j_1}, \ldots, C_{j_{k-1}}}      \max_{M_{j_1}, \ldots, M_{j_{k-1}}}  \mathop{\mathbb{E}}_{C_{j_k}}  \left( I(A : B_{C_{j_1}})_{\pi_C} + \ldots + \max_{M_{j_k}} I(A : B_{C_{j_k}} | B_{C_{j_1}}, \ldots, B_{C_{j_{k-1}}})_{\pi_C} \right) \nonumber \\
& =&      \mathop{\mathbb{E}}_{C_{j_1}, \ldots, C_{j_{k-1}}}       \max_{M_{j_1}, \ldots, M_{j_{k-1}}}  \left( I(A : B_{C_{j_1}})_{\pi_C} + \ldots + \mathop{\mathbb{E}}_{C_{j_k}}  \max_{M_{j_k}} I(A : B_{C_{j_k}} | B_{C_{j_1}}, \ldots, B_{C_{j_{k-1}}})_{\pi_C} \right) \nonumber \\
&\geq&      \mathop{\mathbb{E}}_{C_{j_1}, \ldots, C_{j_{k-1}}}         \max_{M_{j_1}, \ldots, M_{j_{k-1}}}  \left( I(A : B_{C_{j_1}})_{\pi_C} + \ldots +  I(A : B_{C_{j_{k-1}}} | B_{C_{j_1}}, \ldots, B_{C_{j_{k-2}}})_{\pi_C}  \right) \nonumber \\
&+&       \mathop{\mathbb{E}}_{C_{j_1}, \ldots, C_{j_{k-1}}}          \min_{M_{j_1}, \ldots, M_{j_{k-1}}}   \mathop{\mathbb{E}}_{C_{j_k}}  \max_{\Lambda_{j_k}} I(A : B_{C_{j_k}} | B_{C_{j_1}}, \ldots, B_{C_{j_{k-1}}})_{\pi_C},  \nonumber \\
&=& f(t-1) +    \mathop{\mathbb{E}}_{C_{j_1}, \ldots, C_{j_{k-1}}}      \min_{M_{j_1}, \ldots, M_{j_{k-1}}}   \mathop{\mathbb{E}}_{C_{j_k}}   \max_{M_{j_k}} I(A : B_{C_{j_k}} | B_{C_{j_1}}, \ldots, B_{C_{j_{k-1}}})_{\pi_C}.  \nonumber
\end{eqnarray}

From Eqs. (\ref{f11}) and (\ref{f22}), we obtain
\begin{equation}
\log d_{A} \geq \sum_{q=1}^k \mathop{\mathbb{E}}_{C_{j_1}, \ldots, C_{j_{q-1}}} \min_{M_{j_1}, \ldots, M_{j_{q-1}}} \mathop{\mathbb{E}}_{C_{j_q}} \max_{M_{j_q}} I(A : B_{C_{j_q}} | B_{C_{j_1}}, \ldots, B_{C_{j_{q-1}}})_{\pi},
\end{equation}
and so there exists a $q \leq k$ such that 
\begin{equation}
\mathop{\mathbb{E}}_{C_{j_1}, \ldots, C_{j_{q-1}}} \min_{M_{j_1}, \ldots, M_{j_{q-1}}} \mathop{\mathbb{E}}_{C_{j}} \max_{M_{j}} I(A : B_{C_j} | B_{C_{j_1}}, \ldots, B_{C_{j_{q-1}}})_{\pi} \leq \frac{\log d_{A}}{t},
\end{equation}
where we relabelled $j_q \rightarrow j$. Thus there exists a $(q-1)$-tuple of sets ${\cal C} := \{ C_{j_1}, \ldots, C_{j_{q-1}} \}$ and measurements $M_{j_1}, \ldots, M_{j_{q-1}}$ such that
\begin{equation} \label{smallMI2}
\mathop{\mathbb{E}}_{C_j \notin {\cal C}} \max_{M_{j}} I(A : B_{C_j} | B_{C_{j_1}}, \ldots, B_{C_{j_{q-1}}})_{\pi} \leq \frac{\log d_{A}}{t}.
\end{equation}

Here we can follow the proof of Theorem \ref{main} without any modifications to obtain that 
\begin{equation} 
\mathop{\mathbb{E}}_{C_j \notin {\cal C}}  \left \Vert \tr_{\backslash C_j} \circ \Lambda -  {\cal E}_{C_j}  \right\Vert_{\Diamond} \leq \sqrt{2\ln(2)\frac{(d_A)^6\log d_{A}}{k}},
\end{equation}
Then
\begin{equation} \label{expectation3}
\mathop{\mathbb{E}}_{C_j}  \left \Vert \tr_{\backslash C_j} \circ \Lambda -  {\cal E}_{C_j}  \right\Vert_{\Diamond} \leq \sqrt{2\ln(2)\frac{(d_A)^6\log d_{A}}{k}} + \frac{2kt}{n}.
\end{equation}
Choosing $k$ to minimize the right-hand side as done in the proof of Theorem \ref{main} and applying Markov's inequality, we obtain the result.
\end{proof}

\subsection{Proof of Proposition \ref{guessin}}

We will make use the following well-known lemma:

\begin{lem} \label{gentle}
(Gentle Measurement \cite{Win99}) Let $\rho$ be a density matrix and $N$ an operator such that $0 \leq N \leq \openone$ and $\tr(N \rho) \geq 1 -\delta$. Then
\begin{equation}
\Vert \rho - \sqrt{N} \rho \sqrt{N} \Vert_1 \leq 2\sqrt{\delta}.
\end{equation}
\end{lem}

\begin{repprop}{guessin}
Let ${\cal E}$ be the channel given by Eq. (\ref{channelE}). Suppose that for every $i = \{1, \ldots, t\}$ and $\delta > 0$,
\begin{equation} \label{availability}
\min_{\rho \in {\cal D}(A)}  p_{\text{guess}}( \{ \tr(M_k \rho), \sigma_{B_{j_i}, k} \}    ) \geq 1 - \delta.
\end{equation}
Then there exists POVMs $\{ N_{B_{j_1}, k} \}, \ldots, \{ N_{B_{j_t}, k} \}$ such that
\begin{equation}
\min_{\rho}  \sum_{k} \tr(M_k \rho)  \tr \left(  \left( \bigotimes_{i} N_{B_{j_{i}}, k} \right) \sigma_{B_{j_1} \ldots B_{j_t}, k}   \right) \geq 1 - 6 t \delta^{1/4}.
\end{equation}
%
\end{repprop}
\begin{proof}
For simplicity we will prove the claim for $t=2$. The general case follows by a similar argument.

Since for $j = \{1, 2\}$, $\min_{\rho \in {\cal D}(A)}  p_{\text{guess}}( \{ \tr(M_k \rho), \sigma_{B_j, k} \}    ) \geq 1 - \delta$, by the minimax theorem \cite{Sion58} it follows that there exists POVMs $\{ N_{B_1, k} \}$, $\{ N_{B_2, k} \}$ on $B_1$ and $B_2$, respectively, such that for $j \in \{1, 2\}$ and all $\rho \in {\cal D}(A)$,
\begin{equation} \label{aux1}
\sum_k \tr(M_k \rho) \tr( N_{B_j, k} \sigma_{B_j, k} ) \geq 1 - \delta.
\end{equation}

Fix $\rho$ and let $X_j:= \{ k : \tr(N_{B_j, k} \sigma_{B_j, k}) \leq 1 - \sqrt{\delta} \}$ for $j = \{1, 2\}$. Then from Eq. (\ref{aux1}),
\begin{equation}
\sum_{k \in X_j} \tr(\rho M_k) \leq \sqrt{\delta}.
\end{equation}

Let $G = X_{1}^{c} \cap X_{2}^c$, with $X_{j}^c$ the complement of $X_j$. Then 
\begin{eqnarray}
&& \sum_{k} \tr(M_k \rho)  \tr \left(  \left( N_{B_1, k} \otimes N_{B_2, k} \right) \sigma_{B_1 B_2, k}   \right)  \\
&\geq& \sum_{k \in G} \tr(M_k \rho)  \tr \left(  \left( N_{B_1, k} \otimes N_{B_2, k} \right) \sigma_{B_1 B_2, k}   \right)   \nonumber \\
&\geq& \sum_{k \in G} \tr(M_k \rho)  \tr( N_{B_1, k} \sigma_{B_1, k} ) \tr( N_{B_2, k} \sigma_{B_2, k} ) - 4 \delta^{1/4}  \nonumber \\
&\geq& (1 - \sqrt{\delta})  \sum_{k \in G} \tr(M_k \rho)  \tr( N_{B_1, k} \sigma_{B_1, k} ) - 4 \delta^{1/4}  \nonumber \\
&\geq& (1 - \sqrt{\delta}) (1 - \delta - 2 \sqrt{\delta}) - 4 \delta^{1/4} \nonumber\\
&\geq& 1 - 12 \delta^{1/4},  \nonumber
\end{eqnarray}
where in the third line we used Lemma \ref{gentle}. In more detail, we have
\begin{eqnarray}  \label{aux69}
\tr \left(  \left( N_{B_1, k} \otimes N_{B_2, k} \right) \sigma_{B_1 B_2, k}   \right) =  \tr(N_{B_1, k} \sigma_{B_1, k} )   \tr(N_{B_2, k} \sigma'_{B_2, k} ),
\end{eqnarray}
with $\sigma'_{B_2, k} := \tr_{B_1} (N_{B_1, k} \sigma_{B_1 B_2, k}) / \tr(N_{B_1, k} \sigma_{B_1, k})$. Since $\tr(N_{B_1, k} \sigma_{B_1, k}) \geq 1 - \delta^{1/2}$, Lemma \ref{gentle} gives $\Vert \sigma'_{B_2, k}  - \sigma_{B_2, k}   \Vert_1 \leq 4 \delta^{1/4}$. Then from Eq. (\ref{aux69}),
\begin{equation}
\tr \left(  \left( N_{B_1, k} \otimes N_{B_2, k} \right) \sigma_{B_1 B_2, k}   \right)  \geq \tr(N_{B_1, k} \sigma_{B_1, k} )   \tr(N_{B_2, k} \sigma_{B_2, k}) - 4 \delta^{1/4}.
\end{equation}
\end{proof}

\subsection{Proof of Corollary \ref{discord}}

Corollary \ref{discord} will follow from Theorem \ref{main} and the following well-known continuity relation for mutual information:

\begin{lem}\label{lem:alickifannes} (Alicki-Fannes Inequality~\cite{alickifannes})
For $\rho_{AB}$,
\begin{equation}
\label{eq:alickifannes} 
|H(A|B)_\rho - H(A|B)_\sigma|\leq 2 \| \rho-\sigma\|_1 \log d_A + 2{h_2}(2 \| \rho-\sigma\|_1),
\end{equation}
with $H(A | B) = S(AB) - S(B)$ and $h_2$ the binary entropy function. 

If $S(A)_\rho = S(A)_\sigma$, then
\begin{equation}
\label{eq:alickifannes2}
| I(A:B)_\rho - I(A:B)_\sigma| \leq 2 \| \rho-\sigma\|_1 \log d_A + 2{h_2}(2 \| \rho-\sigma\|_1).
\end{equation} 
\end{lem}

\begin{repcor}{discord}
Let $\Lambda : {\cal D}(B) \rightarrow {\cal D}(B_1 \otimes \ldots \otimes B_n)$ be a cptp map. Define $\Lambda_j :=  \tr_{\backslash B_j} \circ \Lambda$ as the effective dynamics from ${\cal D}(B)$ to ${\cal D}(B_j)$. Then for every $\delta > 0$ there exists a set $S \subseteq [n]$ with $|S| \geq (1-\delta)n$ such that for all $j \in S$ and all states $\rho_{AB}$ it holds
\begin{equation}
I(A:B_j)_{\id_{A}\ot\Lambda_{j}(\rho_{AB})} \leq\max_{\Lambda \in \text{QC}} I(A:B)_{\id \otimes \Lambda(\rho_{AB})} + \epsilon \log d_A + 2 h_2 \left(\epsilon\right),
\end{equation}
where $\epsilon =  2 \left(  \frac{ 27 \ln(2) (d_{B})^6 \log(d_B)}{n\delta^3} \right)^{1/3}$, $h_2$ is the binary entropy function, and the maximum on the right-hand side is over quantum-classical channels $\Lambda(X) = \sum_{l} \tr(N_{l}X) \ket{l}\bra{l}$, with $\{N_l\}_l$ a POVM and $\{\ket{l}\}_l$ a set of orthogonal states.

As a consequence, for every $\rho_{AB}$,
\begin{equation}
\label{eq:asymptoticbroadcastrest}
\lim_{n\rightarrow\infty}\max_{\Lambda_{B\rightarrow B_1B_2\ldots B_n}}\mathop{\mathbb{E}}_j I(A:B_j) =  \max_{\Lambda \in \text{QC}}I(A:B)_{\id \otimes \Lambda(\rho_{AB})},
\end{equation}
with $\mathop{\mathbb{E}}_j X_j = \frac{1}{n}\sum_{i=1}^N X_j$, and the maximum on the left-hand side taken over any quantum operation $\Lambda : {\cal D}(B) \rightarrow {\cal D}(B_1 \otimes \ldots \otimes B_n)$.

\end{repcor}

\begin{proof}
By definition, for all cptp maps $\Lambda$ and $\mathcal{E}$ acting on $B$, and for any state $\rho_{AB}$, it holds
\[
\|\id_{A}\otimes\Lambda_{B}(\rho) - \id_{A}\otimes\mathcal{E}_{B}(\rho) \|_1\leq \|\Lambda - \mathcal{E} \|_\diamond. 
\]
Combining Theorem~\ref{main} and Lemma~\ref{lem:alickifannes} (specifically, Eq.~\eqref{eq:alickifannes2}), we have that for every $\delta > 0$ there exist a measurement $\{ M_{k} \}_k$ and a set $S \subseteq [n]$ with $|S| \geq (1-\delta)n$ such that for all $j \in S$ and all states $\rho_{A'A}$ it holds
\begin{equation}
\label{eq:upperboundmutual}
I(A:B)_{\id_{A}\ot\Lambda_{j}(\rho_{AB})} \leq I(A:B)_{\id_{A}\otimes \mathcal{E}_j (\rho_{AB})} + \epsilon \log d_A + 2 h_2 \left(\epsilon\right),
\end{equation}
with 
\begin{equation}
\mathcal{E}_j(X) = \sum_k \tr(M_k X) \ket{k}\bra{k} 
\end{equation}
and  
\begin{equation}
\epsilon = 2 \left(  \frac{27 \ln(2) (d_{B})^6 \log(d_B)}{n\delta^3} \right)^{1/3}. 
\end{equation}
The claim is then a simple consequence of substituting $\mathcal{E}_j$ with an optimal quantum-classical channel.

We now turn to the proof of Eq.~\eqref{eq:asymptoticbroadcastrest}. That the left-hand side of~Eq.~\eqref{eq:asymptoticbroadcastrest} is larger than the right-hand side is trivial. Indeed one can pick $\Lambda=\Lambda_{B\rightarrow B_1B_2\ldots B_n}$ as the quantum-classical map that uses the POVM $\{N_l\}_l$ that achieves the accessible information $I(A:B_c):=\max_{\Lambda \in \text{QC}}I(A:B)_{\id \otimes \Lambda(\rho_{AB})}$ with measurement on $B$ and stores the result in $n$ classical registers, one for each $B_i$: $\Lambda(X) = \sum_l \tr(N_l X) \ket{l}\bra{l}^{\otimes n}$. To prove that the left-hand side of~Eq.~\eqref{eq:asymptoticbroadcastrest} is smaller than the right-hand side it is sufficient to use Eq.~\eqref{eq:upperboundmutual} for the choice $\delta=n^{-\frac{1-\eta}{3}}$, for any $0<\eta<1$. Then one obtains,
\begin{equation}
\begin{split}
\frac{1}{n}\sum_{i=1}^n I(A:B_i)&\leq\frac{1}{n}\left\{ (1-\delta)n \left[I(A:B_c) + \epsilon \log d_A + 2 h_2 \left(\epsilon\right)\right] + \delta n 2 \log d_A\right\}\\
	&=(1-\delta)  \left[I(A:B_c) + \epsilon \log d_A + 2 h_2 \left(\epsilon\right)\right] + \delta 2 \log d_A\xrightarrow[n\rightarrow\infty]{} I(A:B_c)
\end{split}
\end{equation}
where we have used that 
\begin{equation}
\epsilon =  2 \left(  \frac{27 \ln(2) (d_{B})^6 \log(d_B)}{n\delta^3} \right)^{1/3} \xrightarrow[n\rightarrow\infty]{} 0
\end{equation}
for our choice of $\delta$, independently of the choice of $\Lambda=\Lambda_{B\rightarrow B_1B_2\ldots B_n}$.
\end{proof}

\end{document}